\documentclass[conference]{IEEEtran}

\usepackage{amsmath, amsthm, amssymb}
\usepackage{latexsym}
\usepackage{graphicx}
\usepackage{verbatim}
\usepackage{mathrsfs}
\usepackage{algorithmic}
\usepackage{float}
\usepackage[lofdepth, lotdepth]{subfig}
\usepackage{array}
\usepackage{url}
\usepackage{blkarray, multirow}
\usepackage{color}
\usepackage{enumerate}

\theoremstyle{plain}
\newtheorem{theorem}{Theorem}
\newtheorem{lemma}{Lemma}

\newtheorem{corollary}{Corollary}
\newtheorem{proposition}{Proposition}

\theoremstyle{definition}
\newtheorem{definition}{Definition}
\newtheorem{example}{Example}
\newtheorem{remark}{Remark}

\newcommand{\C}{{\mathcal C}}

\newcommand{\K}{{\mathcal K}}

\newcommand{\N}{{\mathcal N}}

\newcommand{\bM}{{\boldsymbol{M}}}

\newcommand{\bP}{{\boldsymbol P}} 
\newcommand{\bQ}{{\boldsymbol Q}}

\newcommand{\bA}{{\boldsymbol A}}
\newcommand{\bB}{{\boldsymbol B}}

\newcommand{\bU}{\boldsymbol{U}}
\newcommand{\bV}{\boldsymbol{V}}

\newcommand{\bu}{{\boldsymbol u}}
\newcommand{\bv}{{\boldsymbol v}}

\newcommand{\bc}{{\boldsymbol c}}

\newcommand{\bzero}{{\boldsymbol 0}}

\newcommand{\bx}{{\boldsymbol{x}}}

\newcommand{\bG}{{\boldsymbol{G}}}

\newcommand{\bbF}{{\mathbb F}}

\newcommand{\ff}{\mathbb{F}}

\newcommand{\fq}{\mathbb{F}_q}

\newcommand{\fkE}{{\mathfrak E}}

\newcommand{\fkF}{{\mathfrak F}}

\newcommand{\supp}{{\sf supp}}

\newcommand{\dist}{{\mathsf{d}}}

\newcommand{\cA}{{\mathscr{A}}}
\newcommand{\cB}{{\mathscr{B}}}

\newcommand{\cC}{{\mathscr{C}}}
\newcommand{\cG}{{\mathscr{G}}}
\newcommand{\cV}{{\mathscr{V}}}
\newcommand{\cE}{{\mathscr{E}}}

\newcommand{\define}{\stackrel{\mbox{\tiny $\triangle$}}{=}}

\newcommand{\et}{{\emph{et al.}}}

\newcommand{\ledc}{{{$\{N_i,K_i\}_{i=1}^m$-LEDC}}}
\newcommand{\dmax}{d_{\max}}

\title{Locally Encodable and Decodable Codes for Distributed Storage Systems
\thanks{This work was completed when Han Mao Kiah visited Singapore University of Technology and Design.}}
 \author{
   \IEEEauthorblockN{
     Son Hoang Dau\IEEEauthorrefmark{1},
		 Han Mao Kiah\IEEEauthorrefmark{2},
     Wentu Song\IEEEauthorrefmark{3}, 
     Chau Yuen\IEEEauthorrefmark{4}
		} 
   \IEEEauthorblockA{
   \IEEEauthorrefmark{1}\IEEEauthorrefmark{3}\IEEEauthorrefmark{4}Singapore University of Technology and Design,   
   \IEEEauthorrefmark{2}Nanyang Technological University, Singapore\\     	
		Emails: $\{${\it\IEEEauthorrefmark{1}sonhoang\_dau, 
		\IEEEauthorrefmark{2}wentu\_song,
		\IEEEauthorrefmark{3}yuenchau}$\}$@sutd.edu.sg,		
		{\it\IEEEauthorrefmark{2}kiahhanmao@gmail.com}
		}
 }
\IEEEoverridecommandlockouts 
\begin{document}

\maketitle

\begin{abstract}
We consider the locality of encoding and decoding operations 
in distributed storage systems (DSS), and 
propose a new class of codes, called locally encodable and decodable codes (LEDC),
that provides a higher degree of operational locality compared to currently known codes.
For a given locality structure, we derive an upper bound on the global distance and 
demonstrate the existence of an optimal LEDC for sufficiently large field size.
In addition, we also construct two families of optimal LEDC for fields with size linear in code length.
\end{abstract}

\section{Introduction}
\label{sec:intro}

Motivated by practical applications in data storage and communication, 
\emph{locality} has become an increasingly important notion in the study of erasure codes.
A code with some locality property allows certain
operations to be performed by accessing only a portion of a codeword
or a data vector. 
As a result, significant amount of resources (bandwidth, computational power, memory) can be saved. We list below some representative families of
codes with certain locality properties.  
\begin{itemize}
	\item \emph{Locally decodable codes} were proposed by Katz and Trevisan~\cite{KatzTrevisan2000}, which allow a single symbol of the original data to be decoded with high probability by only querying a small number of coded symbols of a possibly corrupted codeword.
\item
\emph{Locally testable code} was first systematically studied by 
Goldreich and Sudan~\cite{GoldreichSudan2002}.
For such a code, there exists a test that checks whether a given string
is a codeword, or rather far from the code, by reading only a constant
number of symbols of the string. 
\item \emph{Locally repairable codes} (LRC) (see e.g.~\cite{OggierDatta2011, Gopalan2012}) were tailored-made for distributed storage systems 
(DSS), which guarantee that any erased coded symbol (corresponding to a failed storage node) can be reconstructed from a small subset of other coded symbols (corresponding to other surviving nodes). 
LRC with a given locality network topology were studied in \cite{Mazumdar2014,Shanmugam2014}. 
\item 
\emph{Update efficient codes}, also known as \emph{locally updatable codes} (see e.g.~\cite{Anthapadmanabhan2010, Chandran2014, Mazumdar2015}), require updating as few nodes
as possible whenever one data symbol is changed.
\item
\emph{Coding with constraints}~\cite{HalbawiHoYaoDuursma2013, HalbawiThillHassibi15, YanSprintsonZelenko2014, 
DauSongYuenISIT14, DauSongYuen_JSAC2014, SongDauYuen15} requires that each coded symbol must be a function of a given subset of data symbols. In a classical code, each coded symbol
can be a function of all data symbols. 
\end{itemize}

\begin{figure}[ht]
\centering
\includegraphics[scale=0.8]{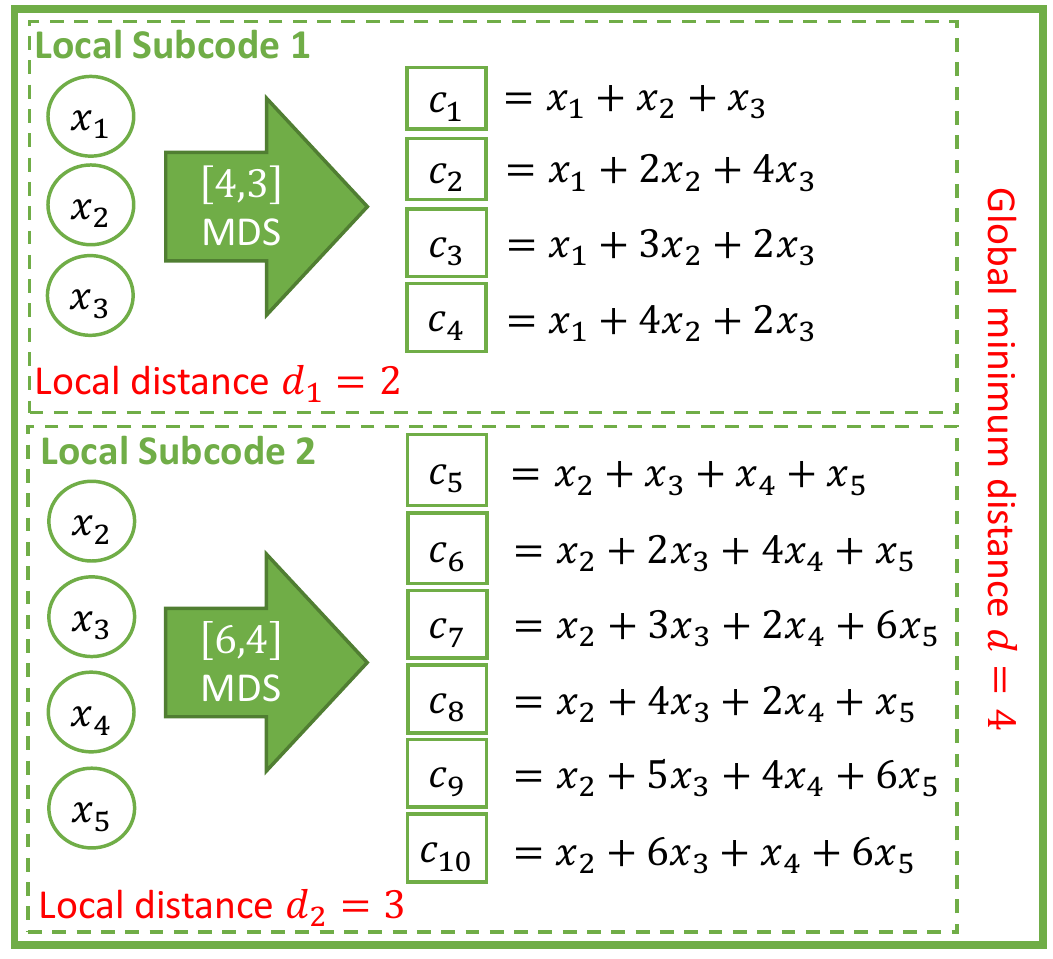}
\caption{An example of an LEDC with two local subcodes. 
Here $\K_1=\{x_1, x_2, x_3\}$, $\K_2 = \{x_2, \ldots, x_5\}$, 
$\N_1=\{c_1,\ldots,c_4\}$, and $\N_2 = \{c_5, \ldots, c_{10}\}$. 
The underlying field is $\ff_7$ - the set of integers modulo $7$. 
Coded symbols in each group
$\N_i$ are encoded from the corresponding data symbols in $\K_i$
and conversely can be used to decode these data symbols.
Each local subcode is an MDS code and moreover, the LEDC, which encodes $k = |\K_1 \cup \K_2|=5$ 
data symbols into $n = |\N_1|+|\N_2| = 10$ coded symbols, can reach the optimal minimum distance $d=4$.
In this example, the subcodes are encoded by simply using two Vandermonde
matrices. In general, this naive construction
yields LEDCs with optimal distances only in certain cases (see 
Theorem~\ref{thm:constructionI} and Example~\ref{ex:counter_ex}).}
\label{fig:toy-example}
\vspace{-15pt}
\end{figure}

Continuing along this line of research, we propose the class of 
\emph{locally encodable and decodable code} (LEDC) that provides
a higher level of operational locality in distributed storage systems (DSS) compared to currently known codes. 
In an LEDC, the set of $n$ coded symbols (corresponding to $n$ storage
nodes) is partitioned into $m$ disjoint subsets $\N_i$'s, each of which is responsible for encoding and decoding of a given subset $\K_i$ of some data symbols.
Each pair $(\N_i,\K_i)$ forms a \emph{maximum distance separable} (MDS) code, referred to as a \emph{local} subcode. 
The parameter of interest is the minimum (Hamming) distance of the
\emph{global} LEDC. An LEDC that has the largest minimum distance
is called \emph{optimal}.
 
We illustrate in Fig.~\ref{fig:toy-example} a toy example
of an optimal LEDC with minimum distance $d = 4$. This LEDC consists of
$m = 2$ local subcodes. The first local subcode, which is a $[4,3]$ MDS
code, encodes three data symbols $\K_1=\{x_1, x_2, x_3\}$ 
into four coded symbols $\N_1=\{c_1,\ldots,c_4\}$.
The second local subcode, which is a $[6,4]$ MDS code, encodes 
four data symbols $\K_2 = \{x_2, \ldots, x_5\}$ into six coded symbols
$\N_2 = \{c_5, \ldots, c_{10}\}$. 
The underlying field is $\ff_7$ - the set of integers modulo $7$.  

An LEDC distinguishes itself from the family of locally repairable codes in the following aspects.
\begin{itemize}
	\item First, \emph{both} encoding and decoding can be done locally 
	for each local subcode of an LEDC - to create coded symbols in $\N_i$, only
	data symbols from $\K_i$ are involved, and conversely, to decode
	data symbols in $\K_i$, only coded symbols in $\N_i$ are involved. 
	In contrast, the encoding operation in each local group in an LRC
	may involve all data symbols and only local repair is required, not
	local decodability.
	In other words, an LRC does not provide local
	encoding and decoding and only guarantees local repair, which can also be done by an LEDC.
	\item Secondly, in the context of LEDC, the structure of the local groups (i.e. $\K_i$'s and $\N_i$'s) are given, 
	whereas in the context of LRC, the size and the repair capability of each group are given as input.
\end{itemize}

LEDC falls under the regime of coding with constraints as each coded symbol must be a function of a given set of data symbols. 
However, while LEDC provides local decodability, coding with 
constraints does not.
% We discuss this issue in more details later. 
Notice also that the notion of \emph{local decodability} in the setting of LEDC is different
from that in the setting of locally decodable codes (LDC)~\cite{KatzTrevisan2000}. Indeed, while an LEDC guarantees that
each given \emph{subset} of data symbols can be decoded (with probability
one) from a specific subset of coded symbol (possibly under some errors/erasures), an LDC requires that \emph{each} data symbol can be decoded (with
high probability) from a small subset of coded symbols (possibly under
some errors). 

A crucial feature of an optimal LEDC is that the repair capability of the global system can be greater than its individual local systems.
To illustrate this, we consider the example in Fig.~\ref{fig:toy-example} and regard the local subcodes as codes operating over independent storage systems.
The first storage system utilizes an MDS code of
minimum distance $d_1 = 2$ and hence tolerates \emph{one} node failure. 
On the other hand, the second storage system utilizes an MDS code of
minimum distance $d_2 = 3$ and hence tolerates \emph{two} node failures. 
If the two codes are co-designed to form an optimal LEDC with optimal distance $d=4$, 
then the LEDC provides extra protection for node failures: 
any \emph{three} node failures across the two systems can be tolerated.

The improvement in fault tolerance results from the fact that the two systems
share some common data symbols ($x_2$ and $x_3$ in this toy example). 
Each system also has some private data symbols ($x_1$ in the first
system and $x_4$ and $x_5$ in the second system). 
In normal condition where the node failures are within the 
local fault tolerance, each storage system can work independently
to repair the failed nodes. No sharing of private data is required. 
However, in a catastrophic scenario where the number
of node failures exceeds the local fault tolerance, the two systems
can cooperate to repair the failed nodes by sharing the private data.

Furthermore, there exist LEDCs whose fault tolerances
exceed the \emph{sum} of the fault tolerances of their local subcodes.

\begin{example}\label{ex:equal}
For example, we can construct an LEDC with the following locality structure:
 \begin{align*}
\N_1 &=\{c_1,c_2,\ldots,c_{5}\}, &
\K_1 &=\{x_1,x_2,x_3,x_{4}\}, \\
\N_2 &=\{c_{6},c_{7},\ldots,c_{12}\}, &
\K_2 &=\{x_2,x_3,\ldots,x_{7}\}. 
\end{align*}
Here, the local subcodes are $[5,4]$ and $[7,6]$ MDS codes that each tolerate up to one erasure.
We demonstrate later in Example~\ref{ex:2} that the LEDC can achieve a minimum distance of five and 
hence is able to tolerate up to four erasures.
\end{example}
%For example, we can construct an LEDC with the following locality structure:
%\begin{align*}
%\N_1 &=\{c_1,c_2,\ldots,c_{20}\}, &
%\K_1 &=\{x_1,x_2,\ldots,x_{16}\}, \\
%\N_2 &=\{c_{21},c_{22},\ldots,c_{41}\}, &
%\K_2 &=\{x_7,x_8,\ldots,x_{23}\}. 
%\end{align*}
%Here, the local subcodes are $[20,16]$ and $[21,17]$ MDS codes that each tolerates up to four erasures.
%We demonstrate that the LEDC can achieve a minimum distance of $15$ and 
%hence is able to tolerate up to 14 erasures.
 
Our key results are summarized below. 
\begin{itemize}
	\item We prove that for any given locality structure (i.e. $\K_i$'s and $\N_i$'s), there always exists an optimal LEDC over any sufficiently large finite fields. The optimal minimum distance, however, can be determined in polynomial time. 
	\item When $m = 2$ (i.e. there are only two local subcodes), we provide constructions for two families of optimal LEDC.
	\begin{enumerate}[(I)]
	\item A straightforward construction using nested MDS codes for the case where $|\K_1\cap \K_2|$ is small.
	\item An algebraic construction of LEDC as a (punctured) subcode of a cyclic code of length $q-1$, where $q$ is the size of the finite field, 
	in the case where $|\N_1-\K_1| = |\N_2 - \K_2|$.
	\end{enumerate}
\end{itemize}

Our paper is organized as follows. 
Necessary definitions and notation are provided in Section~\ref{sec:pre}. 
We prove the existence of optimal LEDC over sufficiently large finite 
fields in Section~\ref{sec:large_field}. 
Section~\ref{sec:cyclic} is devoted for the construction of optimal LEDC 
over small fields when there are two local subcodes. 
%The paper is concluded in Section~\ref{sec:conclusion}.  

\section{Preliminaries}
\label{sec:pre}   

Let $\fq$ denote the finite field with $q$ elements. 
Let $[n]$ denote the set $\{1,2,\ldots,n\}$. 
For a $k \times n$ matrix $\bM$, for $i \in [k]$ and $j \in [n]$, let
$\bM_i$ and $\bM[j]$ denote the row $i$ and the column $j$ of $\bM$, respectively. 
We define below standard notions from coding theory (for instance, see \cite{MW_S}).
 
The \emph{support} of a vector $\bu = (u_1,\ldots,u_n) \in \fq^n$ is the set 
$\supp(\bu) = \{i \in [n] \colon u_i \neq 0\}$. 
The (Hamming) \emph{weight} of a vector $\bu \in \fq^n$ is $|\supp(\bu)|$.
The (Hamming) \emph{distance} between two vectors $\bu$ and $\bv$ of $\fq^n$
is defined to be
$
\dist(\bu,\bv) = |\{i \in [n] \colon u_i \ne v_i\}|. 
$
A $k$-dimensional subspace $\cC$ of $\fq^n$ is called a linear $[n,k,d]_q$ \emph{(erasure) code} over $\fq$ 
if the minimum distance $\dist(\cC)$ between any pair of distinct vectors in $\cC$ is equal to $d$.  
Sometimes we may use the notation $[n,k,d]$ or just $[n,k]$ for the sake of simplicity. The vectors in $\cC$ are called \emph{codewords}. 
It is known that the minimum weight of a nonzero codeword in a linear code $\cC$ is equal to its minimum distance $\dist(\cC)$. 
The well-known Singleton bound (\cite[Ch. 1]{MW_S}) states that for any $[n,k,d]_q$
code, it holds that $d \leq n - k + 1$. 
If the equality is attained, the code is called \emph{maximum distance separable} (MDS).  

A \emph{generator matrix} $\bG$ of an $[n,k]_q$ code $\cC$ is a $k \times n$ matrix whose rows are linearly independent codewords of $\cC$. Then $\cC = \{\bx \bG \colon \bx \in \fq^k\}$. 
It is also well known that if $\dist(\cC) = d$ then $\cC$ can correct any $d-1$ erasures.
In other words, $\bx$ can be recovered from any $n-d+1$ coordinates of the codeword $\bc = \bx\bG$.  
An $[n,k,d]$ code can be used in a DSS as follows.
A vector $\bx$ of $k$ data symbols can be encoded into $n$ coded symbols $\bc = \bx \bG$, 
each of which is stored at one node in the system. 
Then $\bx$ can be recovered from any set of $n-d+1$ nodes. 
Hence the DSS can tolerate $d-1$ node failures.  

\begin{definition}
\label{def:LEDC}
Let $n \geq k \geq 1$ and $m \geq 1$ be some integers.
Let $K_1, \ldots, K_m$ be $m$ nonempty (possibly overlapping)
subsets of $[k]$ such that $[k] = \cup_{i=1}^m K_i$. 
Let $N_1, \ldots, N_m$ be $m$ nonempty non-overlapping subsets of $[n]$ that partition $[n]$, i.e. $N_i \cap N_{i'} = \varnothing$ if $i \neq i'$ and $[n] = \cup_{i=1}^m N_i$. 
Suppose that $n_i = |N_i| \geq |K_i| = k_i$ for every $i \in [m]$. 
An $\{N_i,K_i\}_{i=1}^m$-LEDC over an alphabet $\Sigma$ is a
mapping 
$\fkE: \Sigma^k \longrightarrow \Sigma^n$
that maps a vector of $k$ data symbols $\bx = (x_1,\ldots,x_k) \in \Sigma^k$ into a vector of $n = \sum_{i=1}^m n_i$ coded symbols 
$\bc = (c_1, \ldots, c_n) \in \Sigma^n$
and satisfies the following properties.
\begin{itemize}
	\item[(P1)] 
	The coded symbols in $\N_i \triangleq \{c_j: j \in N_i\}$
	only depends on the data symbols in $\K_i \triangleq \{x_j: j \in K_i\}$
	$(\forall i \in [m])$. 
	\item[(P2)] 
	The set of data symbols $\K_i$ can be
	determined from any subset of $k_i$ coded symbols of the set $\N_i$
	$(\forall i \in [m])$.  
\end{itemize}
\end{definition}

When $\Sigma$ is a finite field $\fq$ for some prime power $q$ and
the mapping $\fkE$ is linear, the corresponding LEDC is called \emph{linear}.
For linear LEDC, the mapping $\fkE$ can be represented by
a $k \times n$ matrix $\bG$ over $\fq$ such that 
$\fkE(\bx) = \bx \bG$. Such a matrix $\bG$ is referred to
as a \emph{generator matrix} of the LEDC.   
In this work we are only interested in linear LEDC. 
The second property (P2) in Definition~\ref{def:LEDC}, which $\fkE$ must satisfy, states that each pair of $n_i$ coded symbols in $\N_i$
and $k_i$ data symbols in $\K_i$ forms a linear $[n_i,k_i]$
MDS code. We refer to these $m$ MDS codes as the \emph{local subcodes} of
the LEDC. 
An example of a linear LEDC with two local subcodes over $\ff_7$ is given in Fig.~\ref{fig:toy-example}. 

\section{Existence of Optimal Locally Encodable and Decodable Codes
Over Large Fields}
\label{sec:large_field}

We first discuss the closely related concept of coding with constraints. 
The upper bound on the minimum distance for a code with coding constraints~\cite{HalbawiThillHassibi15, SongDauYuen15} still applies in the setting of LEDC. 
The existence proof for optimal codes over large finite fields~\cite[Lemma 12]{SongDauYuen15}, however, needs to be appropriately modified to take into account the new locality feature of LEDC. 

\subsection{Coding With Constraints}
\label{subsec:coding_constraints}

In the setting of linear coding with constraints~\cite{HalbawiHoYaoDuursma2013, HalbawiThillHassibi15, YanSprintsonZelenko2014, 
DauSongYuenISIT14, DauSongYuen_JSAC2014, 
SongDauYuen15}, the data vector $\bx \in \fq^k$ is encoded into
the coded vector $\bc = \bx \bG \in \fq^n$ for some $k\times n$
matrix $\bG$ in $\fq$, subjected to the following constraints: 
each coded symbol $c_j$ is a function of a given subset of 
the data symbols indexed by $C_j \subseteq [k]$. 
In a classical code, $C_j \equiv [k]$ for all $j \in [n]$.   
For $i \in [k]$ let $R_i \triangleq \{j: i \in C_j\}$. 
Then it is obvious that the support of the $i$th \emph{row} of any
valid generator matrix $\bG$ of a code with coding constraints
must be included in $R_i$. Similarly, the support of the $j$th
\emph{column} of $\bG$ must be included in $C_j$. 
The following theorem presents an upper bound on the minimum 
distance of a code with coding constraints and states that
an optimal code attaining this upper bound does exist over
a sufficiently large finite field. 

\begin{theorem}(\cite[Corollary 1]{HalbawiThillHassibi15}, \cite[Lemma 12]{SongDauYuen15})
\label{thm:coding_constraints}
Suppose that $\cC$ is a linear code that encodes the vector of data symbols $\bx \in \fq^k$ into the vector of coded symbols $\bc \in \fq^n$
under the following constraints: each $c_j$ is a function of a given subset of 
the data symbols indexed by $C_j \subseteq [k]$. 
Let $R_i \triangleq \{j: i \in C_j\}$. Then
\begin{equation}
\label{eq:d_constr}
\dist(\cC) \leq \dmax \triangleq 1 + \min_{\varnothing \neq I \subseteq [k]}\{|\cup_{i \in I}R_i| - |I|\}.
\end{equation}
Moreover, when $q$ is sufficiently large, there exists a code with minimum distance attaining this bound. 
\end{theorem} 

\begin{proposition}
\label{pro:d-poly}
The upper bound $\dmax$ in Theorem~\ref{thm:coding_constraints} can be found in polynomial time. 
\end{proposition}
\begin{proof}
Note that $\dmax$ is the largest $d$ satisfying
\begin{equation}
\label{eq:3.1}
|\cup_{i \in I} R_i| \geq d - 1 + |I|, \quad \forall \varnothing 
\neq I \subseteq [k]. 
\end{equation}
From the Singleton bound, $d \leq n - k + 1$. 
Hence, we can find $\dmax$ by verifying (\ref{eq:3.1}) for each $d$
ranging from $n-k+1$ down to $1$. As long as (\ref{eq:3.1})
can be verified in polynomial time for every $d \in [n-k+1]$, 
we can find $\dmax$ in polynomial time. 

Note that any $d \in [n-k+1]$ can be written as $d = (n-k+1) - \delta$, 
for some $0 \leq \delta \leq n - k$. Hence, (\ref{eq:3.1}) can 
be rewritten as 
\begin{equation}
\label{eq:3.3}
|\cup_{i \in I} R_i| \geq n - k - \delta + |I|, \quad \forall \varnothing 
\neq I \subseteq [k]. 
\end{equation} 

In the proof of \cite[Lemma 10]{DauSongYuen_JSAC2014}, we provide
a polynomial time algorithm to verify the so-called MDS Condition
\begin{equation}
\label{eq:3.2}
|\cup_{i \in I} R_i| \geq n - k + |I|, \quad \forall \varnothing 
\neq I \subseteq [k],
\end{equation} 
where $R_1,\ldots,R_k$ are arbitrary nonempty subsets of $[n]$. 
We do so by creating a network with one source and $k$ sinks
and prove that (\ref{eq:3.2}) holds if and only if the capacity of
a minimum cut between the source and any sink is at least $n$. 
 
Using exactly the same proof, we can show that (\ref{eq:3.3}) 
holds if and only if the capacity of a
minimum cut between the source and any sink of that network is at least $n - \delta$. 
As the capcity of such a minimum cut can be computed in polynomial time, 
(\ref{eq:3.3}) can be verified in polynomial time and so can (\ref{eq:3.1}). 
\end{proof}

\subsection{Optimal Locally Encodable and Decodable Codes Over
Large Fields}
\label{subsec:large_field}

In this subsection we establish that an optimal LEDC
always exists over a sufficiently large field. 

\begin{theorem}
\label{thm:large_field}
Suppose that $\cC$ is a linear {\ledc}.  
Let $n_i = |N_i|$, $n = \sum_{i=1}^m n_i$, and $k = |\cup_{i = 1}^m K_i|$.
For $i \in [m]$ and $j \in N_i$ let $C_j = K_i$. 
For each $i \in [k]$ let $R_i = \{j: i \in C_j\}$.
Then
\begin{equation} 
\label{eq:d}
\dist(\cC) \leq \dmax \triangleq 1 + \min_{\varnothing \neq I \subseteq [k]}\{|\cup_{i \in I}R_i| - |I|\}.
\end{equation} 
Moreover, when $q$ is sufficiently large, there exists a linear {\ledc}
over $\fq$ with minimum distance attaining this bound. 
\end{theorem} 

The upper bound (\ref{eq:d}) simply follows from Theorem~\ref{thm:coding_constraints}. 
Note that by Proposition~\ref{pro:d-poly}, this upper bound
can be determined in polynomial time. 
We present below a proof of 
the existence of optimal LEDCs over large fields. 
We aim to show that when $q$ is sufficiently large,
there always exists a $k \times n$ matrix $\bG$ over $\fq$
that generates an {\ledc} with minimum distance attaining 
(\ref{eq:d}). 

Firstly, observe that if $\bG$ is a generator matrix of 
an {\ledc} then by (P1), $\supp(\bG_i) \subseteq R_i$ and 
$\supp(\bG[j]) \subseteq C_j$ for all $i \in [k]$ and $j \in [n]$.
Let $\bG^\xi = (g^\xi_{i,j})_{k\times n}$ where 
\vspace{-5pt}
\[
g^\xi_{i,j} = 
\begin{cases}
\xi_{i,j}, &\text{if } j \in R_i,\\
0, &\text{otherwise,}
\end{cases}
\]
where $\xi_{i,j}$'s are indeterminates. 

For any subset $J \subseteq [n]$ of size $n - \dmax + 1$
let $\cG^J = (\cV,\cE)$ be the bipartite graph defined
as follows. The vertex set $\cV$ can be partitioned into two
parts, namely, the left part ${\sf{L}} = \{\ell_1, \ldots, \ell_k\}$, and the right part 
${\sf{R}} = \{{\sf{r}}_j: j \in J\}$. The edge set is  \vspace{-5pt}
\[ 
\cE = \big\{ (\ell_i, {\sf{r}}_j) \colon i \in [k], \ j \in J \cap R_i\big\}. 
\]
Then for every $\varnothing \neq I \subseteq [k]$, the neighbor set of $\{\ell_i: i \in I\}$ has size at least
\[
\begin{split}
|N(\{\ell_i: i \in I\})| &= |\cup_{i \in I} (J\cap R_i)|\\ 
&= |J \cap \big(\cup_{i \in I} R_i\big)|\\
&= |\big(\cup_{i \in I} R_i\big) \setminus 
\big([n] \setminus J \big)|\\
&\geq |\big(\cup_{i \in I} R_i\big)| - (n-|J|)\\
&\geq (\dmax + |I| - 1) - (n - (n-\dmax+1))\\
&=|I|.
\end{split}
\] 
Hence, according to the famous Hall's marriage theorem, 
there is a matching of size $k$ in $\cG^J$. In other words,
there exist $k$ distinct indices $j_1, \ldots, j_k$ in $J$ such that
$j_i \in R_i$ for all $i \in [k]$.  

For each subset $J \subseteq [n]$ of size $n - \dmax + 1$, 
we consider the submatrix $\bP^J$ of $\bG^\xi$ that consists
of the columns indexed by $j_1, \ldots, j_k$ as discussed above. 
Then the determinant $\det(\bP^J)$, which is a multivariable polynomial
in $\fq[\ldots, \xi_{i,j},\ldots]$, is not identically zero. 
The reason is that in the expression of $\det(\bP^J)$ as a sum
of monomials, there is one monomial that cannot be canceled out, 
namely $\prod_{i=1}^k \xi_{i,j_i}$. 
Let 
\[
\fkF^{\text{dist}} = \prod_{\genfrac{}{}{0pt}{}{J \subseteq [n]}{|J| = n-\dmax+1}}
\det(\bP^J) \in \fq[\ldots, \xi_{i,j},\ldots]. 
\]
Then $\fkF^{\text{dist}}$ is not identically zero. 
Roughly speaking, the polynomial $\fkF^{\text{dist}}$ captures the locality structure $\{N_i,K_i\}_{i=1}^m$ of the LEDC and the desired minimum distance $\dmax$. 

Regarding the locality for decoding, i.e. every $k_i$
coded symbols in $\N_i$ can be used to recover all data symbols in 
$\K_i$, let 
\[
\fkF^{\text{MDS}} = \prod_{i=1}^m 
\prod_{\bQ_i} \det(\bQ_i) \in \fq[\ldots, \xi_{i,j}, \ldots],
\]
where the second product is taken over all $k_i \times k_i$ submatrices
$\bQ_i$ that consist of $k_i$ rows of $\bG^\xi$ indexed by $K_i$
and some $k_i$ columns of $\bG^\xi$ among those indexed by $N_i$. 
By definition of $\bG^\xi$, each of such matrices $\bQ_i$ has a
nontrivial determinant. Therefore, $\fkF^{\text{MDS}}$ is not identically zero. 

Thus, 
\[
\fkF = \fkF^{\text{dist}} \times \fkF^{\text{MDS}} \not\equiv \bzero.
\] 
Therefore, by \cite[Lemma 4]{Ho2006}, 
for sufficiently large $q$, there exist $g_{i,j} \in \fq$
such that $\fkF(\ldots,g_{i,j}, \ldots,) \neq 0$. Let $\bG = (g_{i,j})$
(for $i,j$ where $j \notin R_i$, simply let $g_{i,j} = 0$). 
Since $\fkF^{\text{dist}}(g_{i,j}) \neq 0$, the linear code generated
by $\bG$ has minimum distance at least $\dmax$. 
Moreover, since $\fkF^{\text{MDS}}(g_{i,j}) \neq 0$, 
the coded symbols in each $\N_i$ form an $[n_i,k_i]$ MDS code. 
Hence, $\bG$ generates an optimal LEDC. 
We complete the proof of Theorem~\ref{thm:large_field}. 

\section{Optimal Locally Encodable and Decodable Codes With Two Local Subcodes}
\label{sec:cyclic}

In this section we restrict ourselves to the case $m=2$, i.e.
when the LEDC has exactly two local subcodes. We provide two 
constructions of optimal LEDCs over fields of sizes linear in $n$. 

Without loss of generality, we consider  the following locality structure:
\begin{align*}
K_1&=\{1,2,\ldots,k_1\},\ K_2 =\{k_1-t+1,k_1-t+2,\ldots,k\},\\
N_1&=\{1,2,\ldots,n_1\},\ N_2 =\{n_1+1,n_1+2,\ldots,n\}.
\end{align*}
Here $t= |K_1 \cap K_2|$ is the number of common data symbols
shared by the two local subcodes.
For brevity purpose, we denote an LEDC with this locality structure as an $[n_1,k_1;n_2,k_2;t]$-LEDC.

%When $r = 0$, the two subcodes share no common data symbols. 
%Then as long as each local subcode is an MDS code (e.g. a Reed-Solomon 
%code), the LEDC is trivially optimal with minimum distance $d = 1+\min\{n_1-k_1,n_2-k_2\}$. 

When $t = k$, i.e. $K_1 = K_2 = [k]$, there is no locality in 
encoding, as both subcodes use all of the data symbols in their
encoding process. An optimal LEDC is simply an $[n,k]$ MDS code
with minimum distance $d = n - k + 1$. 
In the remainder of this section, we always assume that $t < k$. 
The following is a straightforward corollary of Theorem~\ref{thm:large_field}.
 
\begin{corollary}
Suppose that $t = |K_1 \cap K_2| < k$. 
Additionally, assume that either $t < \min\{k_1,k_2\}$ or $n_1-k_1 = n_2 - k_2$.
If there exists a linear $[n_1,k_1;n_2,k_2;t]$-LEDC with minimum
distance $d$ then 
\begin{equation}
\label{eq:d2}
d \leq 1 + t + \min\{n_1 - k_1, n_2 - k_2\}.
\end{equation}
\end{corollary}

The upper bound \eqref{eq:d2} reflects the fact that the more common data 
the local subcodes have, the larger the global minimum distance of the LEDC. In one extreme case when $t = 0$, i.e. the two subcodes share no common data symbols, the global minimum distance is equal to one local minimum distance, whichever smaller. 
Hence, the global LEDC offers no further protection against erasures to each local subcode. When $t$ is sufficiently large, however, the
global LEDC can offer considerable amount of additional protection
against erasures. 

We henceforth assume, without loss of generality, that $n_1 -k_1
\leq n_2 - k_2$. Then \eqref{eq:d2} can be rewritten as
\begin{equation}
\label{eq:d2_reduced}
d \leq 1 + t + n_1 - k_1.
\end{equation}
We aim to construct optimal LEDCs whose minimum distances meet the upper bound \eqref{eq:d2_reduced}.
%Specifically, we construct $[n_1,k_1;n_2,k_2;t]$-LEDC whose minimum distance is $n_1-k_1+t+1$.

In what follows, we consider a linear $[n_1,k_1;n_2,k_2;t]$-LEDC $\cC$ and 
describe $\cC$ via its generator matrix $\bG$.
From the local encoding property, we may write $\bG$ as
%\[\left(\begin{array}{c|c} \bU &\bzero \\\hline \bA &\bB\\\hline \bzero & \bV\end{array}\right),\]
\begin{equation}
\label{mat:G}
\begin{blockarray}{ccl}
\ n_1 & \ n_2 & \\
\begin{block}{@{\hspace*{2pt}}(c|c@{\hspace*{2pt}})l}
 \ \bU \quad & \bzero \ & \quad k_1-t \\
	\cline{1-2}
 \ \bA & \bB \ & \quad t \\
	\cline{1-2}
 \ \bzero & \bV \ & \quad k_2-t \\
\end{block}
\end{blockarray}
\vspace{-5pt}
\end{equation}
\noindent where $\bU$, $\bA$, $\bB$, and $\bV$ are $(k_1-t)\times n_1$, 
 $t\times n_1$, $t\times n_2$, and $(k_2-t)\times n_2$ matrices, respectively.
From the properties of an LEDC, the linear code $\cA$ with generator matrix 
$\bG_\cA\triangleq\left(\begin{array}{c} \bU \\\hline \bA \end{array}\right)$ 
is an $[n_1,k_1]$ MDS code and
the linear code $\cB$ with generator matrix 
$\bG_\cB\triangleq\left(\begin{array}{c} \bB \\\hline \bV \end{array}\right)$ 
is an $[n_2,k_2]$ MDS code.

\subsection{Optimal LEDCs from Nested MDS codes}

Our first construction uses pairs of nested MDS codes.
More formally, for $1\leq k'<k\leq n$, the pair $(\cC',\cC)$ is a pair of nested $(k',k;n)$ MDS codes
if $\cC'\subset\cC$ and $\cC'$ and $\cC$ are $[n,k']$ and $[n,k]$ MDS codes, respectively.

\begin{theorem}[Construction I] 
\label{thm:constructionI}
Suppose $n_2-k_2+1\geq t$ and 
there exist pairs of nested $(k_1-t,k_1;n_1)$ and $(k_2-t,k_2;n_2)$ MDS codes.
Then there exists an optimal $[n_1,k_1;n_2,k_2;t]$-LEDC $\cC$ whose minimum distance is $n_1-k_1+t+1$.
\end{theorem}

\begin{proof}
Let $(\cA',\cA)$ be the pair of nested  $(k_1-t,k_1;n_1)$ MDS codes with generator matrices $\bU$ and $\left(\begin{array}{c} \bU \\\hline \bA \end{array}\right)$.
Similarly, let $(\cB',\cB)$ be the pair of nested  $(k_2-t,k_2;n_2)$ MDS codes with generator matrices $\bV$ and $\left(\begin{array}{c} \bB \\\hline \bV \end{array}\right)$.

Then we consider a typical nonzero codeword $\bc=(\bx_{1},\bx_{2},\bx_{3})\bG$, where $\bx_1=(x_i:i\in K_1\setminus K_2)$, $\bx_2=(x_i:i\in K_1\cap K_2)$ and
$\bx_3=(x_i:i\in K_2\setminus K_1)$. Hence, $\bc=(\bx_1\bU+\bx_2\bA,\bx_2\bB +\bx_3\bV)$. We have the following cases.
\begin{enumerate}[(i)]
\item When $\bx_2=\bzero$ and $\bx_1\ne \bzero$, then the first $n_1$ coordinates $\bx_1\bU$ of $\bc$ is a nonzero codeword from $\cA'$ and has weight at least $n_1-k_1+t+1$.
\item When $\bx_2=\bzero$ and $\bx_3\ne \bzero$, then the last $n_2$ coordinates $\bx_2\bV$ of $\bc$ is a nonzero codeword from $\cB'$ and has weight at least $n_2-k_2+t+1 \geq n_1 - k_1 + t + 1$.
\item When $\bx_2\ne \bzero$, then $\bc$ consists of nonzero codewords from $\cA$ and $\cB$. Hence, $\bc$ has weight at least $(n_1-k_1+1)+(n_2-k_2+1) \geq n_1 - k_1 + t + 1$, because we assume $n_2-k_2+1\geq t$.
\end{enumerate}
Therefore, the minimum distance of $\cC$ is $n_1-k_1+t+1$, achieving 
the upper bound in \eqref{eq:d2_reduced}.
\end{proof}

Pairs of nested MDS codes were studied by Ezerman~{\et}~\cite{Ezerman.etal:2013} in the context of quantum codes.
Specifically, pairs of nested MDS codes of all possible parameters 
(assuming the MDS conjecture) were constructed and 
in general, a pair of $(k',k;n)$ MDS codes exists if $n\leq q$.
Therefore, $\max\{n_1,n_2\}\leq q$ suffices for Construction I.
The LEDC in Fig~\ref{fig:toy-example} is, in fact, constructed
using two pairs of nested codes over $\ff_7$, generated by 
Vandemonde matrices. Note that a Cauchy matrix also generates a pair of nested MDS codes, while requires a larger field size than a Vandermonde matrix. 

When the assumption in Theorem~\ref{thm:constructionI} does not hold, 
i.e. $t > 1+\max\{n_1-k_1, n_2-k_2\}$, Construction I may fail to produce
an optimal LEDC. We demonstrate this fact below. 

\begin{example} 
\label{ex:counter_ex}
Let $K_1 = \{1,2,3,4\}$, $K_2 = \{2,3,4,5\}$, and $n_1 = n_2 = 5$. 
In this case, $t = |K_1 \cap K_2| = 3 > 1+\max\{n_1-k_1,n_2-k_2\}$. 
Hence, the assumption of Theorem~\ref{thm:constructionI} is violated. 
In fact, using two Vandermonde matrices over $\ff_7$ for the local subcodes
results in the following generator matrix:
\[
\bG = 
\begin{pmatrix}
1 & 1 & 1 & 1 & 1 & 0 & 0 & 0 & 0 & 0\\
1 & 2 & 3 & 4 & 5 & 1 & 1 & 1 & 1 & 1\\
1 & 4 & 2 & 2 & 3 & 1 & 2 & 3 & 4 & 5\\
1 & 1 & 6 & 1 & 6 & 1 & 4 & 2 & 2 & 3\\
0 & 0 & 0 & 0 & 0 & 1 & 1 & 6 & 1 & 6 
\end{pmatrix}.
\]
This matrix generates an LEDC of distance $d = 4$, which does not reach 
the upper bound \eqref{eq:d2_reduced}. Hence, it is not optimal.  
\end{example}

\subsection{Algebraic Construction of Optimal LEDCs}

We aim to construct optimal $[n_1,k_1; n_2, k_2;t]$-LEDCs whenever $n_1+n_2 \leq q-1$ and $r\triangleq n_1-k_1=n_2-k_2$.
Such LEDCs have minimum distance achieving the bound \eqref{eq:d2_reduced}, i.e. $d = r + t + 1$. 
The generator matrices $\bG$ of such optimal LEDCs
have the block form given in \eqref{mat:G}.  
In what follows, we find appropriate field elements $\{u_j,v_j: 0\leq j\leq r+t-1\}$, $\{a^{(\ell)}_j: 1\leq \ell\leq t, 0\leq j\leq r+t-\ell \}$, $\{b^{(\ell)}_j: 1\leq \ell\leq t, 0\leq j\leq r+ \ell-1\}$
and construct four component matrices of $\bG$, namely $\bU$, $\bV$, $\bA$, and $\bB$ 
in the following forms. 

%%%%%%%%%%%%%%%%%%
{\small
\[
\begin{split}
\bU &=
\left(\begin{array}{ccc ccc ccc ccc}
u_0 & u_1  & \cdots  & \cdots & u_{r+t} & 0 &   \cdots & 0 \\
0 & u_0  & u_1  & \cdots  & u_{r+t-1} & u_{r+t} &  \cdots & 0 \\
\vdots & \vdots   & \ddots & \ddots   & \ddots & \ddots & \ddots & \vdots \\
0 & 0 & \cdots  & u_0  & \cdots & \cdots   & u_{r+t-1} & u_{r+t}
\end{array}\right),\\
\bA &=
\left(\begin{array}{cc cccc cccccc}
 0 & 0 & \cdots & 0 & a_0^{(1)} & a_1^{(1)} & \cdots  & \cdots & \cdots & a_{r+t-1}^{(1)} \\
 0 & 0 & \cdots & 0 & 0 & a_0^{(2)} & \cdots  & \cdots & \cdots & a_{r+t-2}^{(2)} \\
\vdots & \vdots & \ddots & \vdots & \vdots& \vdots & \ddots& \ddots& \ddots& \vdots \\
 0 & 0 & \cdots & 0 & 0 & 0 & \cdots & a_0^{(t)}  & \cdots  & a_{r}^{(t)} \\
\end{array}\right),\\
\bB &=
\left(\begin{array}{cc cccc cccccc}
0 & 0 & \cdots & 0 & 0 &  \cdots & 0 & b_0^{(1)} & \cdots & b_{r}^{(1)}\\ 
0 & 0 & \cdots & 0 & 0 & \cdots & b_0^{(2)} & b_1^{(2)} & \cdots & b_{r+1}^{(2)}\\ 
\vdots & \vdots & \reflectbox{$\ddots$} & \vdots & \vdots& \reflectbox{$\ddots$} & \reflectbox{$\ddots$}& \reflectbox{$\ddots$}& \reflectbox{$\ddots$}& \vdots \\
0 & 0 & \cdots & 0 & b_0^{(t)} & \cdots & \cdots & \cdots & \cdots & b_{r+t-1}^{(t)}\\ 
\end{array}\right),\\
\bV &=
\left(\begin{array}{ccc ccc ccc ccc}
0 & 0 & \cdots  & v_0  & \cdots & \cdots & v_{r+t-1} & v_{r+t}\\
\vdots & \vdots  & \reflectbox{$\ddots$} & \reflectbox{$\ddots$} & \reflectbox{$\ddots$}   & \reflectbox{$\ddots$} & \reflectbox{$\ddots$}      & \vdots \\
0 & v_0 & v_1 & \cdots   & v_{r+t-1} & v_{r+t} &   \cdots & 0 \\
v_0 & v_1 & \cdots & \cdots & v_{r+t} & 0  &  \cdots & 0 \\
\end{array}\right).
\end{split}
\]
}
In this subsection, we demonstrate the existence of field elements
$u_j$'s, $v_j$'s, $a_j^{(\ell)}$, and $b_j^{(\ell)}$ such that the above construction yields a generator matrix $\bG$ of an optimal LEDC.
In particular, we prove the following theorem.

\begin{theorem}[Construction II] 
\label{thm:small_field}
Suppose $n_1-k_1=n_2-k_2=r$ and $n_1+n_2\leq q-1$. 
Then there exists an optimal $[n_1,k_1;n_2,k_2;t]$-LEDC $\cC$ whose minimum distance is $r+t+1$.
\end{theorem}

%First, we observe that if $g_0$, $h_0$, $a^{(\ell)}_0$ and $b^{(\ell)}_0$ are nonzero for $1\leq \ell\leq r$, 
%then the matrices $\bG_\cC$, $\bG_\cA$ and $\bG_\cB$ have full rank. 
%Therefore, the linear codes $\cC$, $\cA$ and $\cB$ have the desired dimensions.
To derive the distance properties, our strategy is to show that 
the linear codes $\cA$, $\cB$, and $\cC$ are \emph{subcodes} of certain cyclic codes~\cite[Chap. 7]{MW_S}.
To this end, suppose that $q-1 \geq n_1+n_2$ and identify a vector
 $(c_0,c_1,\ldots,c_{q-2})\in \bbF_q^{q-1}$ with the polynomial $\sum_{j=0}^{q-2}c_j x^j$ 
in the quotient ring $\bbF_q[x]/\langle x^{q-1}-1\rangle$.
%However, note that our LEDC $\cC$ is of length $n = n_1 + n_2 \leq q-1$. 
%Its two local subcodes $\cA$ and $\cB$ have lengths $n_1$ and $n_2$, 
%respectively, and both are strictly smaller than $q - 1$. 
%Note, however, that while the cyclic codes that we employ are of length $q-1$ over $\fq$, our codes $\cA$, $\cB$, and $\cC$ may have shorter lengths, namely $n_1$, $n_2$, and $n$, respectively, and are obtained by puncturing subcodes 
%of these cyclic codes. In other words, redundant coordinates of these subcodes that are always zeros are deleted to produce $\cA$, $\cB$, and $\cC$, without affecting their minimum distances. 
Note, however, that while the cyclic codes that we consider are of length $q-1$ over $\fq$, our codes $\cA$, $\cB$, and $\cC$ may have shorter lengths, namely $n_1$, $n_2$, and $n$, respectively. 
However, we can regard $\cA$, $\cB$, and $\cC$ as codes of length $q-1$ by simply
adding a sufficient number of zero coordinates to the right of each codeword.
Doing so obviously does not affect the polynomial representation of each codeword.  
Thus, from now on we can treat $\cA$, $\cB$, and $\cC$ as subspaces of $\fq^{q-1}$. 
 
Under the mapping of vectors to polynomials, let $(u_0,u_1,\ldots, u_{r+t})$ and $(v_0,v_1,\ldots, v_{r+t})$ correspond to $u(x)$ and $v(x)$, which are polynomials of degree at most $r+t$, respectively.
For $1\leq \ell\leq t$, let $\left(a_0^{(\ell)}, a_1^{(\ell)}, \ldots, a_{r+t-\ell}^{(\ell)}\right)$ and
$\left(b_0^{(\ell)}, b_1^{(\ell)}, \ldots, b_{r+\ell-1}^{(\ell)}\right)$ correspond to
$a^{(\ell)}(x)$ and $b^{(\ell)}(x)$, which are polynomials of degrees at most $r+t-\ell$ and $r+\ell-1$, respectively.

Furthermore, the codewords in the linear codes $\cC$, $\cA$, and $\cB$,
in their polynomial representations, can be obtained via linear combinations of 
certain polynomials described by $u(x)$, $v(x)$, $a^{(\ell)}(x)$, and $b^{(\ell)}(x)$.
Specifically,
\begin{enumerate}
\item $\cC$ is the vector subspace of $\fq^{q-1}$ spanned by the polynomials
$x^iu(x)$ for $0\leq i\leq k_1-t-1$,
$x^jv(x)$ for $n_1\leq j\leq n_1+k_2-t-1$, and
$c^{(\ell)}(x) \define x^{k_1-t+\ell-1}a^{(\ell)}(x)+x^{n_1+k_2-\ell}b^{(\ell)}(x)$ for $1\leq \ell\leq t$,

\item $\cA$ is the vector subspace of $\fq^{q-1}$ spanned by the polynomials
$x^iu(x)$ for $0\leq i\leq k_1-t-1$ and
$x^{k_1-t+\ell-1}a^{(\ell)}(x)$ for $1\leq \ell\leq t$,

\item $\cB$ is the vector subspace of $\fq^{q-1}$ spanned by the polynomials
$x^jv(x)$ for $n_1\leq j\leq n_1+k_2-t-1$ and
$x^{n_1+k_2-\ell}b^{(\ell)}(x)$ for $1\leq \ell\leq t$.
\end{enumerate}

Next, we provide sufficient conditions for 
the polynomials $u(x)$, $v(x)$, $a^{(\ell)}(x)$, and $b^{(\ell)}(x)$
so that the codes $\cC$, $\cA$, and $\cB$ have the desired dimension and distance properties.

\begin{proposition}\label{prop:poly}
Let $\omega$ be a primitive element in $\bbF_q$. 
Suppose that the polynomials $u(x)$, $v(x)$, $a^{(\ell)}(x)$, and $b^{(\ell)}(x)$, $1\leq \ell\leq t$,
satisfy the following conditions.
 \begin{enumerate}[(D1)]
 \item $u_0$, $v_0$, $a^{(\ell)}_0$ and $b^{(\ell)}_0$ are nonzero for $1\leq \ell\leq t$.
 \item $u(\omega^j)=v(\omega^j)=0$ for $0\leq j\leq r+t-1$.
 \item For $1\leq \ell\leq t$, $a^{(\ell)}(\omega^j)=b^{(\ell)}(\omega^j)=0$ for $0\leq j\leq r-1$.
 \item For $1\leq \ell\leq t$, $c^{(\ell)}(\omega^j)=0$ for $0\leq j\leq r+t-1$, 
 where
 $c^{(\ell)}(x)=x^{k_1-t+\ell-1}a^{(\ell)}(x)+x^{n_1+k_2-\ell}b^{(\ell)}(x)$.
 \end{enumerate}
Then the LEDC $\cC$ defined as above is an optimal $[n_1,k_1;n_2,k_2;t]$-LEDC of minimum distance $r+t+1$.
\end{proposition}

To prove this proposition, we employ the well-known BCH bound on minimum distance of a cyclic code. 
We first recall some necessary notations from ~\cite[Ch. 7]{MW_S}.
Let $\C$ be a cyclic code of length $q-1$ over $\fq$. 
An element $z \in \fq$ is called a \emph{zero} of $\C$
if $c(z) = 0$ for every codeword $c(x) \in \C$. 
Let $Z$ be the set of all zeros of $\C$. %Then $u(x) \in \C$ if and only if $u(z) = 0$ for all $z \in Z$.  
The polynomial $g(x) \triangleq \prod_{\alpha \in Z} (x-\alpha)$ is called the \emph{generator polynomial} of $\C$. 
Then $c(x) \in \C$ if and only if $g(x) | c(x)$. 
%A cyclic code is uniquely determined by its generator polynomial.

\begin{theorem}[BCH bound] Let $\omega$ be the primitive element of $\bbF_q$ and $r$ be an integer.
The cyclic code with the generator polynomial $(x-1)(x-\omega)\cdots(x-\omega^{r-1})$
has minimum distance at least $r+1$.
\end{theorem}
%\begin{theorem}[BCH bound] Let $\omega$ be the primitive element of $\bbF_q$ and $r$ be an integer.
%Any cyclic code that has $1, \omega, \omega^2, \ldots, \omega^r$ as zeros must have the minimum distance 
%at least $r+1$.
%\end{theorem}

\begin{proof}[Proof of Proposition \ref{prop:poly}]
First observe that from condition (D1), the matrices $\bG$ given in \eqref{mat:G}, 
$\bG_\cA = \left(\begin{array}{c} \bU \\\hline \bA \end{array}\right)$, 
and $\bG_\cB = \left(\begin{array}{c} \bB \\\hline \bV \end{array}\right)$
all have full rank. Therefore, the corresponding linear codes $\cC$, $\cA$, and $\cB$ have the desired dimensions.
For the distance properties, we make the following two claims.
Recall that we may regard $\cC$, $\cA$, and $\cB$ as codes of length $q-1$ by 
adding a sufficient number of zero coordinates to the right of each codeword. 
\begin{enumerate}[(C1)]
\item $\cC$ is a subcode of the cyclic code with generator polynomial 
$g_1(x) = (x-1)(x-\omega^{})\cdots(x-\omega^{r+t-1})$. 
%More specifically, $\cC$ is obtained by deleting the rightmost $q-1-n$ coordinates of this subcode, which are always zero. 
\item $\cA$ and $\cB$ are subcodes of the cyclic code with generator polynomial 
$g_2(x) = (x-1)(x-\omega^{})\cdots(x-\omega^{r-1})$. 
%These subcodes are also punctured by deleting their rightmost
%$q-1-n_1$ and $q-1-n_2$ coordinates, which are always zero. 
\end{enumerate}

Then by the BCH bound, the codes $\cC$, $\cA$, and $\cB$ have minimum distance $r+t+1$, $r+1$ and $r+1$, respectively and the proposition is immediate. 
%Note that since the puncturing coordinates are always zero in the subcodes, the puncturing operation does not reduce the minimum distance. 

Hence, it suffices to prove the claim. 
%Let $p_1(x)\triangleq(x-1)(x-\omega^{})\cdots(x-\omega^{r+t-1})$.
From conditions (D2) and (D4), we deduce that
\[
u(\omega^j)=v(\omega^j)=c^{(1)}(\omega^j)=\cdots=c^{(t)}(\omega^j), 
\ 0\leq j\leq r+t-1.
\]
Therefore, $g_1(x)=(x-1)(x-\omega^{})\cdots(x-\omega^{r+t-1})$ divides all polynomials in the basis of $\cC$. Hence, (C1) follows.

%To prove (C2), we similarly define $p_2(x)\triangleq(x-1)(x-\omega^{})\cdots(x-\omega^{r-1})$.
It can be also verified that $g_2(x)= (x-1)(x-\omega^{})\cdots(x-\omega^{r-1})$ divides all polynomials in the bases
of $\cA$ and $\cB$, respectively. Hence, (C2) follows. 
\end{proof}
\vskip 5pt 

The following proposition shows that the polynomials satisfying conditions (D1)-(D4)
in Proposition~\ref{prop:poly} do exist. Hence, Theorem~\ref{thm:small_field} follows.

\begin{proposition}\label{prop:equal}
There exist polynomials $u(x)$, $v(x)$, $a^{(\ell)}(x)$, and $b^{(\ell)}(x)$ for $1\leq \ell \leq t$ 
that satisfy conditions (D1)-(D4).
\end{proposition}

To prove this proposition, we consider the following lemma.

\begin{lemma}
\label{lem:equal} 
Fix $1\leq \ell\leq t$, $r\geq 1$ and $t-\ell<T<q-\ell$.
Then there exists polynomials 
$a(x)$ with degree at most $t-\ell$ and
$b(x)$ with degree at most $\ell-1$,
having nonzero constants, such that
\begin{equation}\label{eq:equal}
a(\omega^j)+x^Tb(\omega^j)=0 \mbox{ for }r\leq j\leq r+t-1.
\end{equation}
\end{lemma}

\begin{proof}
Let $a(x)=\sum_{j=0}^{t-\ell}a_jx^j$ and $b(x)=\sum_{j=0}^{\ell-1}b_jx^j$.
Then \eqref{eq:equal} is equivalent to the linear system of equations
$\bM(a_0 , a_1 , \ldots , a_{r-\ell}, b_0 , b_1 , \ldots , b_{\ell-1})^T=\bzero$, 
where $\bM$ is the $t\times (t+1)$ matrix:

{\small
\begin{equation*}\label{linsys:equal}
%\left(\begin{array}{cccc cccc}
%1 & \omega^t & \cdots & (\omega^t)^{r-\ell} & (\omega^t)^T & (\omega^t)^{T+1} & \cdots & (\omega^t)^{T+\ell-1}\\
%1 & \omega^{t+1} & \cdots & (\omega^{t+1})^{r-\ell} & (\omega^{t+1})^T & (\omega^{t+1})^{T+1} & \cdots & (\omega^{t+1})^{T+\ell-1}\\
%\vdots & \vdots & \ddots & \vdots & \vdots & \vdots & \ddots & \vdots \\
%1 & \omega^{r+t-1} & \cdots & (\omega^{r+t-1})^{r-\ell} & (\omega^{r+t-1})^T & (\omega^{r+t-1})^{T+1} & \cdots & (\omega^{r+t-1})^{T+\ell-1}
%\end{array}\right)
\left(\begin{array}{cccc cccc}
1   & \cdots & (\omega^r)^{t-\ell} & (\omega^r)^T &  \cdots & (\omega^r)^{T+\ell-1}\\
1  & \cdots & (\omega^{r+1})^{t-\ell} & (\omega^{r+1})^T &  \cdots & (\omega^{r+1})^{T+\ell-1}\\
\vdots & \ddots & \vdots & \vdots  & \ddots & \vdots \\
1  & \cdots & (\omega^{r+t-1})^{t-\ell} & (\omega^{r+t-1})^T & \cdots & (\omega^{r+t-1})^{T+\ell-1}
\end{array}\right).
%\left(\begin{array}{c}
%a_0 \\ a_1 \\ \vdots \\ a_{r-\ell}\\ b_0 \\ b_1 \\ \vdots \\ b_{\ell-1}
%\end{array}\right)
%=
%\left(\begin{array}{c}
%0 \\ 0 \\ \vdots \\ 0 \\ 0 \\ 0 \\ \vdots \\ 0
%\end{array}\right). 
\end{equation*}
}
%Let $M$ be the $r\times (r+1)$ matrix representing the linear system of equations.
Then $\bM$ can be rewritten as
%\begin{equation*}
%\left(\begin{array}{cccc cccc}
%1 & \omega^t & \cdots & (\omega^{r-\ell})^{t} & (\omega^T)^{t} & (\omega^{T+1})^{t} & \cdots & (\omega^{T+\ell-1})^{t}\\
%1 & \omega^{t+1} & \cdots & (\omega^{r-\ell})^{t+1} & (\omega^{T})^{t+1} & (\omega^{T+1})^{t+1} & \cdots & (\omega^{T+\ell-1})^{t+1}\\
%\vdots & \vdots & \ddots & \vdots & \vdots & \vdots & \ddots & \vdots \\
%1 & \omega^{r+t-1} & \cdots & (\omega^{r-\ell})^{r+t-1} & (\omega^{T})^{r+t-1} & (\omega^{T+1})^{r+t-1} & \cdots & (\omega^{T+\ell-1})^{r+t-1}
%\end{array}\right).
%\end{equation*}
{\small
\begin{equation*}
\left(\begin{array}{cccc cccc}
1 & \cdots & (\omega^{t-\ell})^{r} & (\omega^T)^{r}  & \cdots & (\omega^{T+\ell-1})^{r}\\
1  & \cdots & (\omega^{t-\ell})^{r+1} & (\omega^{T})^{r+1} &  \cdots & (\omega^{T+\ell-1})^{r+1}\\
\vdots  & \ddots & \vdots & \vdots  & \ddots & \vdots \\
1  & \cdots & (\omega^{t-\ell})^{r+t-1} & (\omega^{T})^{r+t-1} &  \cdots & (\omega^{T+\ell-1})^{r+t-1}
\end{array}\right).
\end{equation*}
}
Since $t-\ell<T$ and $T+\ell-1 < q - 1$, the values $1$, \ldots, $\omega^{t-\ell}$, $\omega^{T}$, \ldots, $\omega^{T+\ell-1}$ are distinct nonzero elements in $\bbF_q$. Therefore $\bM$ is the generator matrix of a $[t+1,t,2]$ MDS code.  
Hence, $(a_0,a_1,\ldots,a_{t-\ell}, b_0,b_1,\ldots,b_{\ell-1})$ belongs to its dual code.
Since the dual code is a $[t+1,1,t+1]$ MDS code, which has minimum distance $t+1$,
we can choose $(a_0,a_1,\ldots,a_{t-\ell}, b_0,b_1,\ldots,b_{\ell-1})$ such that $a_0,a_1,\ldots,a_{t-\ell}, b_0,b_1,\ldots,b_{\ell-1}$ are nonzero.
In particular, $a_0$ and $b_0$ are nonzero. Thus, $a(x)$ and $b(x)$ are the desired polynomials.
\end{proof}

\begin{proof}[Proof of Proposition \ref{prop:equal}]
Clearly, setting $u(x)=v(x)=(x-1)(x-\omega^{})\cdots(x-\omega^{r+t-1})$ satisfies conditions (D1) and (D2).

For $1\leq \ell\leq t$, let $a_{*}^{(\ell)}(x)$ and $b_{*}^{(\ell)}(x)$ be the polynomials obtained from Lemma \ref{lem:equal}
by setting $T=(n_1+k_2-\ell)-(k_1-t+\ell-1)$. 
(It is obvious to verify that such $T$ satisfies $t-\ell < T < q - \ell$ for all $1 \leq \ell \leq t$.) 
Then we set
\begin{align*}
a^{(\ell)}(x) &=(x-1)(x-\omega)\cdots(x-\omega^{r-1})a_{*}^{(\ell)}(x),\\
b^{(\ell)}(x) &=(x-1)(x-\omega)\cdots(x-\omega^{r-1})b_{*}^{(\ell)}(x).
\end{align*}
We can then verify that conditions (D1), (D3) and (D4) are satisfied and
this completes the proof.
\end{proof}

\begin{remark}
For $1\leq \ell \leq t$, observe that the polynomials $a_*^{(\ell)}(x)$ and $b_*^{(\ell)}(x)$ are 
found by solving $t$ linear equations in $t+1$ variables in Lemma \ref{lem:equal}.
Therefore, the codes $\cC$, $\cA$ and $\cB$ can be constructed in time polynomial in $q=O(n_1+n_2)$ and $t$.
%\begin{enumerate}
%\item For $1\leq \ell \leq r$, observe that the polynomials $a_*^{(\ell)}(x)$ and $b_*^{(\ell)}(x)$ are 
%found by by solving $r$ linear equations in $r+1$ variables. 
%Therefore, the codes $\cC$, $\cA$ and $\cB$ can be constructed in time polynomial in $q=n_1+n_2+1$ and $r$.
%\item Since $\cC$, $\cA$ and $\cB$ are subcodes of certain cyclic codes, 
%we have efficient decoding algorithms such as the Peterson-Gorenstein-Zierler and Berlekamp-Massey algorithms
%(see for example \cite{Blahut:2003}).
%\end{enumerate}
\end{remark}

\begin{example}[Example~\ref{ex:equal} continued]
\label{ex:2}
We provide the optimal LEDC in Example~\ref{ex:equal} using Construction II.
In this case, $K_1 = \{1,2,3,4\}$, $K_2 = \{2,3,4,5,6,7\}$. Hence, $k_1 = 4$, 
$k_2 = 6$, $k = 7$, and $t = 3$. 
Moreover, $n_1 = 5$, $n_2 = 7$, and hence, $r = n_1 - k_1 = n_2 - k_2 = 1$.  
Let $q = 13$ and $\omega=2$ be the primitive element of $\bbF_{13}$.
Since $r+t-1 = 3$, we set
\[
\begin{split}
u(x)=v(x)&=(x-1)(x-2)(x-2^2)(x-2^3)\\
&= 12+10x+5x^2+11x^3+x^4.
\end{split}
\]
Repeated applications of Lemma~\ref{lem:equal} yield the polynomials
{\small
\begin{align*}
a^{(1)}(x)&= 12+2x+5x^2+7x^3, &
b^{(1)}(x)&= 8+5x, \\
a^{(2)}(x)&= 12+7x+7x^2, &
b^{(2)}(x)&=  7+2x+4x^2, \\
a^{(3)}(x)&=12+x,  &
b^{(3)}(x)&=10+12x+3x^2+x^3.
\end{align*}
}
Hence, the optimal $[5,4;7,6;3]$-LEDC is given by the generator matrix 
\begin{equation*}
\left(\begin{array}{ccccc| ccccccc}
12 & 10 & 5 & 11 & 1 &
0 & 0 & 0 & 0 & 0 & 0 & 0\\ \hline
0 & 12 & 2 & 5 & 7 &
0 & 0 & 0 & 0 & 0 & 8 & 5\\
0 & 0 & 12 & 7 & 7 &
0 & 0 & 0 & 0 & 7 & 2 & 4\\
0 & 0 & 0 & 12 & 1  &
0 & 0 & 0 & 10 & 12 & 3 & 1\\ \hline
0 & 0 & 0 & 0 & 0 &
0 & 0 & 12 & 10 & 5 & 11 & 1 \\
0 & 0 & 0 & 0 & 0 &
0  & 12 & 10 & 5 & 11 & 1 & 0 \\
0 & 0 & 0 & 0 & 0 &
12 & 10 & 5 & 11 & 1 & 0 & 0
\end{array}\right).
\end{equation*}
\end{example}

%\section{Conclusion}
%\label{sec:conclusion} 

\bibliographystyle{IEEEtran}
\bibliography{LEDC}

%\appendix

\end{document}